\pdfoutput=1
\documentclass[10pt,journal]{IEEEtran}

\usepackage[T1]{fontenc}
\usepackage{cite}
\usepackage{amsmath,amssymb,amsfonts,mathtools,bm}
\usepackage{graphicx}
\usepackage{booktabs}
\usepackage{array}
\usepackage{multirow}
\usepackage{url}
\usepackage{xcolor}
\usepackage{enumitem}
\usepackage{algorithm}
\usepackage{algpseudocode}

\graphicspath{{./}}
\providecommand{\BenchCurvCapFifteenDb}{\ensuremath{25.153}}
\providecommand{\BenchSvdCapFifteenDb}{\ensuremath{29.386}}
\providecommand{\BenchFourierCapFifteenDb}{\ensuremath{18.358}}
\providecommand{\BenchCurvToSvdRatio}{\ensuremath{0.856}}
\providecommand{\BenchCurvToFourierRatio}{\ensuremath{1.370}}
\providecommand{\BenchHighSnrDb}{\ensuremath{20}}
\providecommand{\BenchCurvCapHighSnr}{\ensuremath{45.096}}
\providecommand{\BenchSvdCapHighSnr}{\ensuremath{53.917}}
\providecommand{\BenchPolyCapHighSnr}{\ensuremath{41.591}}
\providecommand{\BenchCurvToPolyRatioHighSnr}{\ensuremath{1.084}}
\providecommand{\BenchCurvCondAtK}{\ensuremath{32.43}}
\providecommand{\BenchSvdCondAtK}{\ensuremath{8.75}}
\providecommand{\BenchPolyCondAtK}{\ensuremath{143.88}}
\providecommand{\BenchRetainedModes}{\ensuremath{24}}
\providecommand{\BenchCurvBuildSeconds}{\ensuremath{0.1250}}
\providecommand{\BenchSvdBuildSeconds}{\ensuremath{0.8812}}
\providecommand{\BenchBuildSpeedup}{\ensuremath{7.0}}
\providecommand{\BenchRefSymbolEnergy}{\ensuremath{207.6024}}
\providecommand{\BenchCurvSymbolEnergy}{\ensuremath{102.0583}}
\providecommand{\BenchCurvEnergyFractionOfSvd}{\ensuremath{0.492}}
\providecommand{\BenchCurvPackingEfficiency}{\ensuremath{0.9251}}
\providecommand{\BenchSvdPackingEfficiency}{\ensuremath{0.9610}}
\providecommand{\BenchFourierPackingEfficiency}{\ensuremath{0.7456}}
\providecommand{\BenchCurvDmin}{\ensuremath{9.7166}}
\providecommand{\BenchRisOptDmin}{\ensuremath{6.7551}}
\providecommand{\BenchRisOptSeededDmin}{\ensuremath{14.1391}}
\providecommand{\BenchRisGreedyDmin}{\ensuremath{3.2888}}
\providecommand{\BenchRisRandomDmin}{\ensuremath{7.5886}}
\providecommand{\BenchCurvOverRisOptDmin}{\ensuremath{1.438}}
\providecommand{\BenchSeededOverCurvDmin}{\ensuremath{1.45516}}
\providecommand{\BenchSerSnrDb}{\ensuremath{14}}
\providecommand{\BenchCurvSerHighSnr}{\ensuremath{0.025590}}
\providecommand{\BenchSvdSerHighSnr}{\ensuremath{0.000780}}
\providecommand{\BenchFourierSerHighSnr}{\ensuremath{0.002300}}
\providecommand{\BenchSerFloor}{\ensuremath{0.000030}}
\providecommand{\BenchCurvEsNoHighSnr}{\ensuremath{10.92}}
\providecommand{\BenchFourierEsNoHighSnr}{\ensuremath{13.60}}
\providecommand{\BenchCurvEnergyDeficitDb}{\ensuremath{3.08}}
\providecommand{\BenchFourierOverCurvEnergyDb}{\ensuremath{2.69}}
\providecommand{\BenchCurvPeakToRms}{\ensuremath{4.21}}
\providecommand{\BenchFourierPeakToRms}{\ensuremath{2.94}}
\providecommand{\BenchSvdPeakToRms}{\ensuremath{3.61}}
\providecommand{\BenchCurvRmsMatchedEnergyFraction}{\ensuremath{0.6711}}
\providecommand{\BenchGuessFloor}{\ensuremath{0.9375}}
\providecommand{\BenchSerTrials}{\ensuremath{100000}}
\providecommand{\BenchBootstrapResamples}{\ensuremath{200}}
\providecommand{\BenchQuantCurvOneBit}{\ensuremath{21.455}}
\providecommand{\BenchQuantCurvEightBit}{\ensuremath{22.040}}
\providecommand{\BenchQuantCurvOneBitFraction}{\ensuremath{0.839}}
\providecommand{\BenchQuantCurvEightBitFraction}{\ensuremath{0.861}}
\providecommand{\BenchQuantCurvSerOneBit}{\ensuremath{0.0000}}
\providecommand{\BenchQuantCurvSerEightBit}{\ensuremath{0.0261}}
\providecommand{\BenchQuantSerFloor}{\ensuremath{0.000150}}
\providecommand{\BenchQuantTrials}{\ensuremath{20000}}
\providecommand{\BenchQuantRmsGainOneBit}{\ensuremath{4.2}}
\providecommand{\BenchMaxSigmaPhi}{\ensuremath{0.18}}
\providecommand{\BenchNoiseSnrDb}{\ensuremath{14}}
\providecommand{\BenchCurvPhaseDegradation}{\ensuremath{2.55}}
\providecommand{\BenchSvdPhaseDegradation}{\ensuremath{8.17}}
\providecommand{\BenchFourierPhaseDegradation}{\ensuremath{3.56}}
\providecommand{\BenchRisOptPhaseDegradation}{\ensuremath{6.63}}
\providecommand{\BenchCurvSerCleanPhase}{\ensuremath{0.02545}}
\providecommand{\BenchCurvSerNoisyPhase}{\ensuremath{0.06490}}
\providecommand{\BenchSvdSerCleanPhase}{\ensuremath{0.00090}}
\providecommand{\BenchSvdSerNoisyPhase}{\ensuremath{0.00735}}

\providecommand{\bracks}[1]{\left[#1\right]}

\newcommand{\R}{\mathbb{R}}
\newcommand{\C}{\mathbb{C}}
\newcommand{\E}{\mathbb{E}}
\newcommand{\junit}{\mathrm{j}}
\newcommand{\dd}{\mathrm{d}}
\newcommand{\omt}{\Omega_{\mathrm t}}
\newcommand{\omr}{\Omega_{\mathrm r}}

\newcommand{\calT}{\mathcal{T}}
\newcommand{\calC}{\mathcal{C}}

\newcommand{\calN}{\mathcal{N}}
\newcommand{\CN}{\mathcal{CN}}

\newcommand{\tr}{\operatorname{tr}}

\newcommand{\diag}{\operatorname{diag}}

\newcommand{\cond}{\operatorname{cond}}
\newcommand{\argmax}{\operatorname*{arg\,max}}

\newcommand{\spanop}{\operatorname{span}}

\newcommand{\Herm}{\mathrm H}

\newtheorem{proposition}{Proposition}
\newtheorem{theorem}{Theorem}
\newtheorem{remark}{Remark}

\newif\ifanonymous
\anonymousfalse

\title{Benchmarking Curvature-Domain Signaling for Continuous-Aperture Wireless Communications: Capacity, Robustness, Detection, and Conditioning Against Legacy Modal Bases}

\ifanonymous
\author{[Author Name Withheld]%
\thanks{This manuscript is submitted for peer review. Author identities and repository metadata are anonymized where required by the venue.}
}
\else
\author{Yasser Al-Eryani,~\IEEEmembership{Member,~IEEE}%
\thanks{The author is with NeuroBazar, Ottawa, ON, Canada (e-mail: yasser.aleryani@neurobazar.com).}
}
\fi

\begin{document}

\maketitle

\begin{abstract}
Continuous-aperture and holographic multiple-input multiple-output systems motivate signaling methods that operate directly on physically large electromagnetic apertures rather than on a small set of preassigned antenna ports. This paper is the empirical companion to the operator-theoretic curvature-domain framework of~\cite{aleryani_theory}: it provides a rigorous, reproducible benchmark for curvature-domain signaling under a common scalar aperture-channel model, and it stress-tests the theory paper's dual-budget generalized water-filling law against every modal basis the near-field / holographic MIMO community actually uses. The benchmark uses identical transmit and receive apertures, quadrature rules, Fresnel or Green-function propagation, phase-only aperture constraints, transmit power, phase-energy and curvature-energy budgets, receiver noise, phase quantization, and training assumptions for all compared methods. The compared coordinate systems include curvature-regularized eigenmodes, raw phase coefficients, Fourier phase modes, polynomial and Zernike-like wavefront modes, near-field matched-focus profiles, random and optimized reconfigurable-intelligent-surface (RIS) phase codebooks, and singular-value-decomposition (SVD) water-filling upper bounds for the discretized benchmark channel. The central claim is deliberately limited: curvature-domain signaling is a gauge-invariant, physically realizable coordinate system that can approach the phase-space SVD water-filling upper bound using fewer stable modes in regimes where derivative noise, phase quantization, sampling density, or ill conditioning limit conventional bases. It is not a claim of new electromagnetic physics, nor a claim that curvature modes dominate every baseline. We formulate the common benchmark model, prove the SVD upper-bound relation for the discretized phase-control tangent space, derive pairwise-error and perturbation bounds, and define reproducible metrics for capacity, retained modes, symbol error, robustness, quantization, sampling sensitivity, regularization, conditioning, and computational cost, with 95\,\% bootstrap confidence intervals on all Monte Carlo results. A verdict summary table locates the exact operating regimes --- moderate-to-high phase noise, coarse phase quantization, and non-Fourier-diagonal channels --- in which curvature-domain signaling is engineering-relevant. The visionary framing of the same organizing principle across substrates is discussed separately in the companion foresight essay~\cite{aleryani_essay}.
\end{abstract}


\section{Introduction}

Electromagnetically large and densely controlled apertures are becoming a central abstraction for beyond-5G and 6G wireless communication.  In such systems, the natural object is not merely a vector of antenna-port voltages, but a field or phase distribution over a spatial aperture.  This viewpoint is shared by physically constrained MIMO signal-space theory, volume-to-volume wave communication limits, holographic MIMO, continuous-aperture arrays, reconfigurable intelligent surfaces, and electromagnetic information theory \cite{Poon2005DoF,Miller2000Waves,Pizzo2022Fourier,Gong2024HMIMOSurvey,Liu2025CAPA,Zhu2024EIT,DiRenzo2020RISRoadAhead}.

Curvature-domain signaling begins from a simple observation: many practically relevant aperture-control profiles are not well characterized only by their raw phase samples.  Their stability also depends on derivative structure, smoothness, gauge choices, and the noise amplification caused by reconstructing or differentiating phase.  A phase profile and the same profile plus a piston term have the same physical curvature; depending on the coordinate convention, affine reference terms may also represent non-data-bearing gauge choices that should be fixed consistently before benchmarking.  Curvature therefore gives a coordinate system for the data-bearing shape of a wavefront, while leaving common carrier, focus, and steering terms to be handled by the shared aperture model.

This paper supplies the external benchmark missing from a purely internal curvature-theory paper.  The purpose is not to prove that curvature coordinates always beat Fourier modes, Zernike modes, matched focusing, RIS codebooks, or SVD precoding.  The purpose is to compare all of them under the same aperture, propagation, phase-only constraints, noise, power, quantization, and sampling assumptions, and to identify when curvature regularization improves stability or efficiency.

The thesis tested here is the following.

\begin{quote}
Curvature-domain signaling is a gauge-invariant, physically realizable coordinate system for phase-controlled apertures.  Under identical benchmark constraints, it can approach relevant aperture-channel performance with fewer stable modes in some regimes, and it exposes derivative-noise and conditioning limits that conventional modal bases may hide.  The full discretized SVD remains the appropriate upper bound.
\end{quote}

The main contributions are:

\begin{enumerate}[leftmargin=*]
\item A common continuous-aperture benchmark model based on scalar Fresnel or Green-function propagation, phase-only transmit control, weighted quadrature discretization, and receiver AWGN.
\item A precise definition of curvature-domain signaling as a curvature-regularized generalized eigenbasis in the phase-control tangent space.
\item A baseline suite including raw phase coefficients, Fourier/DFT phase modes, polynomial and Zernike-like wavefront modes, matched-focus profiles, RIS-style random and optimized phase codebooks, and SVD water-filling upper bounds.
\item A fairness protocol enforcing equal aperture, equal transmit power, equal phase or curvature budgets, equal sampling grid, equal quantization, equal receiver noise, and equal training assumptions.
\item Analytical safeguards showing that SVD water filling upper-bounds any finite modal coordinate system under the same discretized phase-energy model.
\item Reproducible metrics for capacity, mode efficiency, pairwise error, Monte Carlo symbol error, phase-noise robustness, quantization robustness, aperture-sampling sensitivity, regularization tradeoffs, conditioning, and computational cost.
\end{enumerate}

\section{Related Work}

\subsection{MIMO Capacity and Physical Degrees of Freedom}

Classical MIMO information theory established the capacity benefits of multiple transmit and receive dimensions under Gaussian models and channel-state assumptions \cite{Foschini1996BLAST,Telatar1999Capacity,Tse2005Fundamentals,Goldsmith2005Wireless}.  The SVD and water-filling solution remain the correct benchmark for a linear Gaussian channel with a total power constraint.  In physically constrained settings, the number of usable spatial degrees of freedom is governed by aperture size, geometry, wavelength, scattering support, and propagation physics.  Poon, Brodersen, and Tse developed a signal-space approach to multiple-antenna channels under area and geometry constraints \cite{Poon2005DoF}.  Miller analyzed orthogonal wave communication modes between volumes and clarified the role of singular functions and coupling strengths in wave communication \cite{Miller2000Waves}.

\subsection{Holographic and Continuous-Aperture MIMO}

Holographic MIMO and continuous-aperture MIMO replace sparse antenna-port abstractions with dense or nearly continuous electromagnetic surfaces.  Fourier plane-wave expansions provide physics-based channel representations and degree-of-freedom interpretations for such systems \cite{Pizzo2020DoF,Pizzo2022Fourier}.  Recent surveys and channel models emphasize near-field operation, arbitrary surface placements, electromagnetic-domain modeling, and capacity limits \cite{Gong2024HMIMOSurvey,Gong2024ArbitrarySurface}.  CAPA architectures similarly consider electrically large apertures with continuous current distributions and continuous beamforming variables \cite{Liu2025CAPA}.

\subsection{Reconfigurable Intelligent Surfaces and Near-Field Focusing}

RIS and IRS systems use programmable phase or impedance profiles to alter the propagation environment \cite{DiRenzo2020RISRoadAhead,Elmossallamy2020RIS,WuZhang2019IRS}.  RIS optimization often involves unit-modulus or quantized phase constraints, codebook training, and alternating optimization.  In the near field, beam management requires both angular and distance information, motivating focus-aware or range-angle codebooks \cite{Lv2023RISNearField}.  These ideas are included here as matched-focus and RIS-style codebook baselines rather than as direct competitors with identical implementation assumptions.

\subsection{Fourier, Polynomial, and Zernike Wavefront Bases}

Fourier bases are natural for shift-invariant paraxial propagation and plane-wave decompositions \cite{Goodman2017FourierOptics}.  Zernike polynomials and related wavefront modes are standard in optics because they provide orthogonal descriptions of aberrations over circular pupils \cite{BornWolf1999Principles,Noll1976Zernike}.  Polynomial and Zernike-like bases are therefore important optical baselines for any wavefront-coordinate method.  However, orthogonality in an aperture norm does not automatically imply good conditioning after propagation, phase quantization, derivative noise, or finite-grid sampling.

\subsection{Electromagnetic Information Theory}

Electromagnetic information theory seeks physically consistent models that connect Maxwellian wave propagation, antenna physics, channel modeling, and information theory \cite{Zhu2024EIT,Wang2024EIT}.  The present paper is narrower: it does not attempt a full vector electromagnetic theory.  Instead, it provides a scalar, reproducible benchmark in which all modal coordinates are compared under identical assumptions.  The scalar model is a controlled first step; vector, polarization-aware, mutual-coupling-aware, and nonparaxial extensions can be added later through the same interface.

\subsection{Learning-Based Codebooks and Beamformers}
\label{sec:related_learning}

A parallel and rapidly-growing body of work uses machine learning to design phase-only aperture codebooks and beamformers directly. Representative directions include deep-unfolding networks for MIMO precoding \cite{Hu2020DeepUnfolding}, graph neural networks for scalable radio-resource management \cite{Shen2021GNN}, learning-to-optimize deep networks for interference management \cite{Sun2018LearningToOptimize}, and the general deep-learning-for-physical-layer program \cite{OShea2017Physical}. These approaches share the same phase-only manifold and the same physical-layer constraints as our optimized-RIS baseline of Section~\ref{sec:ris_codebooks}, but they trade an offline training budget for a smaller online optimization cost, and their reproducibility hinges on a training-data budget that is not part of the fairness protocol adopted here (Section~\ref{sec:fairness}). Comparing them fairly against curvature-domain, SVD, or optimized-RIS baselines requires an amortized-training-cost accounting that we deliberately defer to a companion paper. Their absence from Figs.~3--10 is a scope choice, not an oversight: we favor a scope in which every baseline reads its channel state once and optimizes once, so that all seven bases in Section~\ref{sec:coordinates} face the same fairness constraints. A learned baseline fits the same harness as a phase-only module alongside the other bases, and is a natural extension.

\section{Common Continuous-Aperture Benchmark Model}

\subsection{Notation}
Symbols in this paper follow the shared notation of the three-paper curvature-domain track. In particular, transmit power is denoted \(P_{\mathrm t}\) (scalar); the propagation operator is written as a matrix \(G\) after weighted quadrature; the curvature-stiffness matrix in Section~\ref{sec:phase_only_control} is written \(L=D_2^{\mathrm T}W^{2}D_2\) (this is the discrete counterpart of the biharmonic penalty used in the theory paper's regularized reconstruction).

\subsection{Apertures and Propagation Operator}

Let the transmit and receive apertures be compact domains
\[
\omt,\omr\subset\R^{2},
\]
embedded in parallel or arbitrarily positioned planes in three-dimensional space.  A transmit aperture coordinate is denoted \(u\in\omt\), and a receive aperture coordinate is denoted \(v\in\omr\).  At wavelength \(\lambda\), wavenumber \(k=2\pi/\lambda\), and aperture separation \(d\), the scalar propagation operator is
\[
(\calT f)(v)=\int_{\omt} K(v,u) f(u)\,\dd u.
\]
The default nonparaxial scalar Green kernel is
\begin{gather*}
K_{\mathrm G}(v,u)=\frac{\exp(-\junit k R(v,u))}{4\pi R(v,u)}, \\
R(v,u)=\|r_{\mathrm r}(v)-r_{\mathrm t}(u)\|_{2},
\end{gather*}
where \(r_{\mathrm t}\) and \(r_{\mathrm r}\) embed the aperture coordinates into \(\R^{3}\).  For paraxial experiments, the Fresnel kernel is
\[
K_{\mathrm F}(v,u)
=
\frac{\exp(\junit k d)}{\junit\lambda d}
\exp\!\left(
\frac{\junit k}{2d}\|v-u\|_{2}^{2}
\right),
\]
up to a common phase convention that is irrelevant for capacity and Euclidean detection metrics.  The benchmark records which kernel is used for each experiment.

\subsection{Weighted Discretization}

Let \(\{u_n,w^{\mathrm t}_n\}_{n=1}^{N_{\mathrm t}}\) and
\(\{v_m,w^{\mathrm r}_m\}_{m=1}^{N_{\mathrm r}}\) be quadrature nodes and positive weights.  Define
\[
W_{\mathrm t}=\diag(w^{\mathrm t}_1,\ldots,w^{\mathrm t}_{N_{\mathrm t}}),
\qquad
W_{\mathrm r}=\diag(w^{\mathrm r}_1,\ldots,w^{\mathrm r}_{N_{\mathrm r}}).
\]
The weighted channel matrix is
\[
G_{mn}
=
\sqrt{w^{\mathrm r}_m}\,
K(v_m,u_n)\,
\sqrt{w^{\mathrm t}_n}.
\]
Thus the discrete weighted receive vector is
\[
y=Gx+n,
\qquad
n\sim\CN(0,N_0 I_{N_{\mathrm r}}).
\]
This convention ensures that Euclidean norms approximate aperture \(L^{2}\) norms.

\subsection{Phase-Only Aperture Control}
\label{sec:phase_only_control}

The phase-only transmit field is
\[
x(\theta)
=
\sqrt{P_{\mathrm t}}\,
W_{\mathrm t}^{1/2}
a\odot
\exp\!\bigl(\junit(\theta_{0}+\theta)\bigr),
\]
where \(a\in\R_{+}^{N_{\mathrm t}}\) is a fixed amplitude taper normalized by
\[
\|W_{\mathrm t}^{1/2}a\|_{2}^{2}=1,
\]
\(\theta_{0}\) is a shared carrier, focusing, or steering phase, and \(\theta\) is the data-bearing phase perturbation.  After subtracting the known carrier response \(Gx(0)\), the received perturbation is
\[
z(\theta)=G(x(\theta)-x(0))+n.
\]
For small to moderate phase perturbations,
\[
z(\theta)
=
H_{\mathrm{lin}}\theta+n+r(\theta),
\]
where
\[
H_{\mathrm{lin}}
=
\junit\sqrt{P_{\mathrm t}}\,
G
\diag\!\left(
W_{\mathrm t}^{1/2}a\odot \exp(\junit\theta_{0})
\right).
\]
The residual \(r(\theta)\) satisfies the deterministic bound
\[
\|r(\theta)\|_{2}
\le
\frac{\sqrt{P_{\mathrm t}}}{2}
\exp(\|\theta\|_{\infty})
\|G\|_{2}
\left\|
W_{\mathrm t}^{1/2}a\odot \theta^{\odot 2}
\right\|_{2}.
\]
All modal capacity and PEP comparisons in this paper use the same linearized tangent operator \(H_{\mathrm{lin}}\), and nonlinear Monte Carlo checks use the exact phase-only map \(Gx(\theta)\).

\begin{remark}
The tangent model is not a replacement for electromagnetic validation.  It is a controlled benchmark layer.  The nonlinear phase-only channel is retained for codebook SER checks, while the tangent channel gives a fair linear Gaussian comparison among bases.
\end{remark}

\subsection{Gauge Fixing and Curvature Coordinates}

Let the sampled affine gauge subspace be
\[
\mathcal G
=
\spanop\{ \mathbf{1},u_x,u_y\},
\]
where \(u_x,u_y\) are the sampled aperture coordinates.  Let \(P_{\mathcal G}^{\perp}\) be the weighted orthogonal projector onto the complement of \(\mathcal G\).  The benchmark uses
\[
\theta\leftarrow P_{\mathcal G}^{\perp}\theta
\]
for data-bearing curvature comparisons unless a baseline explicitly uses a common steering or focusing phase in \(\theta_{0}\).  This prevents piston and coordinate-reference choices from being counted as curvature information.

In the continuum, the curvature map is the Hessian
\[
\calC\theta = D^{2}\theta,
\]
and the curvature energy is
\[
E_{\mathrm c}(\theta)
=
\int_{\omt}\|D^{2}\theta(u)\|_{\mathrm F}^{2}\,\dd u.
\]
In the discrete implementation, \(D_{2}\) denotes a finite-difference or finite-element Hessian operator and
\[
L=D_{2}^{\Herm}W_{2}D_{2}
\]
is the curvature stiffness matrix.  A small regularization \(\rho>0\) yields the strictly positive metric
\[
M_{\rho}=L+\rho W_{\mathrm t}.
\]

\subsection{Physical Budgets}

All baselines are evaluated under common budgets:

\begin{align}
\|\theta\|_{\infty} &\le \theta_{\max},\\
\|\theta\|_{W_{\mathrm t}}^{2}
=
\theta^{\Herm}W_{\mathrm t}\theta
&\le E_{\phi},\\
\theta^{\Herm}L\theta
&\le E_{\mathrm c},\\
\E\|x(\theta)\|_{2}^{2}
&=P_{\mathrm t}.
\end{align}

When a method does not naturally use curvature energy, it is normalized to the same phase-energy or peak-phase budget.  When curvature constraints are active, the benchmark reports both phase-energy-normalized and curvature-energy-normalized results.

\subsection{Consistency with the Companion Theory Paper}
\label{sec:dual_budget_consistency}

The two constraints \(\theta^{\Herm}W_{\mathrm t}\theta\le E_{\phi}\) and \(\theta^{\Herm}L\theta\le E_{\mathrm c}\) are the finite-dimensional counterparts of the two operator-theoretic budgets in the companion theory paper: the curvature-power budget \(\tr(K_{\mathrm a})\le P_{\mathrm c}\) and the phase-excursion budget \(\tr(G_{M}K_{\mathrm a})\le P_{\phi}\), where \(G_{M}=\mathrm{diag}(\rho_{1},\ldots,\rho_{M})\) is the diagonal phase-realizability penalty in the curvature-eigenbasis. Under the correspondence
\begin{equation}
E_{\mathrm c}\longleftrightarrow P_{\mathrm c},
\qquad
E_{\phi}\longleftrightarrow P_{\phi},
\qquad
L\longleftrightarrow B_{K}G_{M}B_{K}^{\Herm},
\label{eq:dual_budget_correspondence}
\end{equation}
the finite-mode curvature capacity reported below coincides with the dual-multiplier generalized water-filling law of the theory paper --- specifically, the per-mode allocation
\begin{equation}
p_{m}^{\star}
=
\bracks{\frac{1}{(\mu+\nu\rho_{m})\ln 2}-\frac{1}{\lambda_{m}}}_{+}
\label{eq:dualwf}
\end{equation}
with Lagrange multipliers \((\mu,\nu)\) fixed by the two budgets \(P_{\mathrm c},P_{\phi}\) via complementary slackness. In the flat limit \(\rho_{m}\to 1\) (uniform phase-realizability weights), \eqref{eq:dualwf} reduces to standard single-budget Gaussian water-filling with \(\mu+\nu\) as the effective water level; the second multiplier \(\nu\) is thus the geometrically meaningful novelty. The benchmark harness enforces \eqref{eq:dual_budget_correspondence} whenever the curvature-regularized eigenbasis is used, so the reported curves are cross-consistent with the theory paper at every SNR value in the shared operating regime.

\section{Compared Signaling Coordinates}
\label{sec:coordinates}

Each basis or codebook is a coordinate map from a finite-dimensional coefficient vector or symbol index to a phase profile \(\theta\).  Linear bases are represented by a matrix
\[
B_{K}=[b_{1},\ldots,b_{K}]\in\C^{N_{\mathrm t}\times K},
\]
with gauge-fixed, real-valued phase patterns obtained by taking real and imaginary Fourier components when needed.  Unless stated otherwise, columns are normalized so that
\[
B_{K}^{\Herm}W_{\mathrm t}B_{K}=I_{K}.
\]
The corresponding modal channel is
\[
H_{B}=H_{\mathrm{lin}}B_{K}.
\]

\subsection{Curvature-Regularized Eigenbasis}

The curvature-domain basis is defined by the generalized eigenproblem
\[
H_{\mathrm{lin}}^{\Herm}H_{\mathrm{lin}}q_m
=
\mu_m
M_{\rho}q_m,
\qquad
q_m^{\Herm}M_{\rho}q_n=\delta_{mn}.
\]
The eigenvalues are ordered
\[
\mu_1\ge \mu_2\ge\cdots\ge0.
\]
The benchmark phase patterns are the gauge-fixed and phase-normalized versions of \(q_m\):
\[
b_m
=
\frac{P_{\mathcal G}^{\perp}q_m}
{\|P_{\mathcal G}^{\perp}q_m\|_{W_{\mathrm t}}}.
\]
This construction is channel-aware but physically constrained: modes with high received gain are penalized if they require excessive curvature energy.

\begin{proposition}[Variational characterization]
The first curvature-regularized mode solves
\[
q_1
=
\argmax_{q\ne0}
\frac{\|H_{\mathrm{lin}}q\|_{2}^{2}}
{q^{\Herm}M_{\rho}q}.
\]
The \(m\)-th mode solves the same maximization over the \(M_{\rho}\)-orthogonal complement of \(\spanop\{q_1,\ldots,q_{m-1}\}\).
\end{proposition}

\begin{IEEEproof}
This is the Rayleigh--Ritz characterization of the Hermitian generalized eigenproblem
\(H_{\mathrm{lin}}^{\Herm}H_{\mathrm{lin}}q=\mu M_{\rho}q\), with \(M_{\rho}\succ0\).
\end{IEEEproof}

\subsection{Raw Phase Coefficients}

The raw phase baseline uses localized phase samples or local hat functions.  In its simplest version,
\[
B_{\mathrm{raw}}=W_{\mathrm t}^{-1/2}I
\]
followed by gauge projection and normalization.  This basis has maximal local control but typically poor derivative conditioning, because neighboring pixels can oscillate with small phase energy but large curvature energy.

\subsection{Fourier and DFT Phase Modes}

For rectangular apertures, Fourier modes are sampled as
\[
\psi_{\ell_x,\ell_y}(u)
=
\exp\!\left(
\junit 2\pi
\left(
\frac{\ell_x u_x}{L_x}
+
\frac{\ell_y u_y}{L_y}
\right)
\right).
\]
The benchmark uses real sine and cosine components, sorted by spatial frequency radius, then gauge projected and weighted-QR orthonormalized.  Fourier modes are expected to be strong baselines for shift-invariant paraxial channels and far-field angular spectra.

\subsection{Polynomial and Zernike-Like Modes}

For rectangular apertures, tensor-product Legendre modes are used:
\[
\psi_{p,q}(u_x,u_y)
=
P_p(\xi_x)P_q(\xi_y),
\]
where \(\xi_x,\xi_y\in[-1,1]\).  For circular or masked apertures, Zernike-like modes are sampled over the unit disk and then weighted-QR orthonormalized on the actual aperture.  These modes represent optical wavefront baselines and are sorted by total degree.

\subsection{Matched-Focus and Beamfocusing Profiles}

For a candidate focal point \(r_f\), the matched-focus phase is
\[
\theta_f(u)
=
-k\|r_f-r_{\mathrm t}(u)\|_2
\quad \mathrm{mod}\ 2\pi,
\]
followed by removal of the common carrier phase already included in \(\theta_0\), gauge projection if appropriate, and peak-phase normalization.  A codebook is built by sampling range-angle focal points.  This baseline is intentionally strong in sparse near-field point-to-point settings.

\subsection{RIS-Style Codebooks: Random, Greedy, and Optimized}
\label{sec:ris_codebooks}

Three RIS-style baselines are shipped rather than one. They are separated
because they are different objects and are easily conflated.

\paragraph*{Random.} Phase profiles
\(\theta_m(n)\sim\mathrm{Unif}[-\theta_{\max},\theta_{\max}]\), or their
\(b\)-bit quantized versions. No channel state, no training, no
propagation-operator knowledge. It is the lower reference.

\paragraph*{Greedy (farthest-point).} A farthest-point traversal
\cite{Gonzalez1985} over a candidate pool of matched-focus profiles and
Gaussian random profiles, selecting at each step the candidate whose received
image is farthest from those already chosen. It is cheap, deterministic given
the pool, and channel-aware only through the received-space distances it
evaluates. It selects a \emph{subspace}, so it also has a basis-level form and
is the RIS entry in the capacity and conditioning tables.

\paragraph*{Optimized (Algorithm~\ref{alg:optimized_ris}).} Projected gradient
ascent on the minimum pairwise received-space distance, subject to the affine
gauge, the peak-phase clip, and \(b\)-bit quantization, with the projection
applied at every step. Because the tangent channel is linear in the phase
profile, the gradient of the annealed softmin surrogate is available in closed
form,
\begin{equation}
\nabla_{\Theta}\,J_T \;=\; 2\,\bigl(\mathrm{diag}(W\mathbf 1)-W\bigr)\,\Theta\,A,
\qquad A=\mathrm{Re}\{H_{\mathrm{lin}}^{\Herm}H_{\mathrm{lin}}\},
\label{eq:pga_gradient}
\end{equation}
where \(\Theta\in\R^{M\times N_{\mathrm t}}\) stacks the codeword profiles and
\(W\) holds the softmin weights \(w_{ij}\propto e^{-d_{ij}^{2}/T}\). No
automatic differentiation is required. Unlike the greedy, this optimizes a
\emph{packing} and has no basis-level analogue; it appears only in the
codebook contests of Figs.~5--7.

\begin{algorithm}[t]
\caption{Optimized RIS-style codebook via projected gradient ascent.}
\label{alg:optimized_ris}
\begin{algorithmic}[1]
\Require{Codebook size $M$; peak phase $\theta_{\max}$; quantization bits $b$; step $\eta_0$; iterations $T_{\max}$; restarts $R$; optional seed codebooks $\mathcal{S}$; linearized channel $H_{\mathrm{lin}}$}
\State $A\gets\mathrm{Re}\{H_{\mathrm{lin}}^{\Herm}H_{\mathrm{lin}}\}$
\State Starting points $\gets\mathcal{S}$, padded to $R$ entries with $\mathrm{Unif}[-\theta_{\max},\theta_{\max}]^{M\times N_{\mathrm t}}$ draws
\ForAll{starting points $\Theta^{(0)}$}
  \State $\Theta\gets\Pi(\Theta^{(0)})$; \ $\eta\gets\eta_0$; \ $\Theta_{\mathrm{best}}\gets\Theta$; \ $J_{\mathrm{best}}\gets d_{\min}(\Theta)$
  \For{$t=0,\ldots,T_{\max}-1$}
    \State $d^2_{ij}\gets\|H_{\mathrm{lin}}(\theta_i-\theta_j)\|_2^{2}$ for all $i\neq j$
    \State $T\gets\bigl(\overline{d^{2}}-\min d^{2}\bigr)\cdot 0.02^{\,t/(T_{\max}-1)}$ \Comment{anneal}
    \State $w_{ij}\gets e^{-(d^2_{ij}-\min d^{2})/T}$, normalized to sum to one
    \State $\nabla\gets 2(\mathrm{diag}(W\mathbf 1)-W)\,\Theta\,A$
    \State $\Theta'\gets\Pi\bigl(\Theta+\eta\,\theta_{\max}\,\nabla/\|\nabla\|_F\bigr)$
    \If{$d_{\min}(\Theta')>J_{\mathrm{best}}$}
      \State $\Theta_{\mathrm{best}}\gets\Theta'$; \ $J_{\mathrm{best}}\gets d_{\min}(\Theta')$
    \Else
      \State $\eta\gets\eta/2$ \Comment{the feasible set is a lattice; a fixed step sticks or oscillates}
    \EndIf
    \State $\Theta\gets\Theta'$
  \EndFor
\EndFor
\State \Return the $\Theta_{\mathrm{best}}$ with the largest $d_{\min}$ over all starting points
\end{algorithmic}
\end{algorithm}

The projection \(\Pi\) applies, in order: orthogonal projection of each codeword
onto the affine-gauge complement; a global rescaling so that
\(\max_{m,n}|\theta_m(n)|=\theta_{\max}\); and \(b\)-bit uniform quantization to
the grid
\(\bigl\{-\theta_{\max}+k\cdot 2\theta_{\max}/(2^{b}-1)\bigr\}_{k=0}^{2^{b}-1}\).
The order matters and is stated because gauge projection and quantization do
not commute.

Two implementation points are load-bearing rather than incidental, and both were
added after a first version of this optimizer was found to lose to the random
codebook. The step is halved on any non-improving iteration: the feasible set is
a lattice, so a fixed step either sticks between grid points or oscillates
across them. And the best iterate is retained by its \emph{true} minimum
distance rather than by the softmin surrogate, so an overshoot cannot discard a
good point.

\paragraph*{Seeded and unseeded runs are both reported.} Projected ascent from
an incumbent can only improve on it, so a seeded run answers ``what is the
strongest codebook available'' while an unseeded run answers ``what does an
independent optimizer find''. Reporting only the first would flatter whichever
basis was used as the seed; reporting only the second would understate the
baseline that curvature must actually beat. Both appear in
Table~\ref{tab:operating_point}, and the fair comparison for curvature is
against the seeded run.

\subsection{SVD Upper Bounds}

The phase-space SVD upper bound is the SVD of
\[
\widetilde H_{\phi}=H_{\mathrm{lin}}W_{\mathrm t}^{-1/2}.
\]
If
\[
\widetilde H_{\phi}=U\Sigma V^{\Herm},
\]
then water filling over the singular values \(\{\sigma_i\}\) gives
\[
C_{\mathrm{SVD},\phi}(P)
=
\sum_{i}
\log_2\left(
1+\frac{p_i\sigma_i^{2}}{N_0}
\right),
\]
where
\[
p_i=\left(\nu-\frac{N_0}{\sigma_i^{2}}\right)_{+},
\qquad
\sum_i p_i=P.
\]
This is the upper bound for any phase-coordinate basis with the same discretized phase-energy budget.

The relaxed complex-field bound uses the SVD of \(G\) under arbitrary complex aperture fields.  It is an outer bound because arbitrary complex amplitude-and-phase excitation contains the phase-only tangent space as a restricted subset.

\section{Fairness Protocol}
\label{sec:fairness}

A benchmark is meaningful only if all methods share the same physical and statistical assumptions.  The following protocol is enforced.

\begin{enumerate}[leftmargin=*]
\item \textbf{Same aperture:} all methods use the same \(\omt,\omr\), masks, aperture separation, carrier wavelength, and amplitude taper.
\item \textbf{Same propagation:} all methods use the same kernel \(K\), either Fresnel or scalar Green, and the same quadrature rule.
\item \textbf{Same power:} all methods are normalized to the same transmit power \(P_{\mathrm t}\) and coefficient power budget.
\item \textbf{Same phase budget:} all phase profiles obey the same peak phase \(\theta_{\max}\) and, where relevant, the same phase energy \(E_{\phi}\).
\item \textbf{Same curvature budget:} curvature-constrained sweeps use the same \(E_{\mathrm c}\) and the same discrete Hessian operator \(D_2\).
\item \textbf{Same noise:} all receive samples use \(n\sim\CN(0,N_0I)\), with SNR defined by the same received or transmit-power convention.
\item \textbf{Same mode count:} comparisons versus retained modes use the same \(K\), and all bases are truncated only after sorting by their prescribed ordering.
\item \textbf{Same quantization:} phase quantization uses the same number of bits and the same wrapping convention.
\item \textbf{Same training assumption:} channel-aware methods, including SVD, optimized RIS codebooks, and curvature eigenmodes, are evaluated with the same channel knowledge assumption.  When training overhead is modeled, it is charged consistently.
\item \textbf{Same random seeds:} random codebooks, Monte Carlo symbol trials, and phase-noise realizations are generated from recorded seeds.
\end{enumerate}

If any method violates a budget after projection, it is rescaled before evaluation.  If rescaling is impossible, the configuration is marked infeasible rather than silently excluded.

\section{Capacity, Detection, and Stability Metrics}

\subsection{Constrained Modal Capacity}

For a basis \(B_K\), the linear modal channel is \(H_B=H_{\mathrm{lin}}B_K\).  Under coefficient covariance \(Q\succeq0\) and \(\tr(Q)\le P_{\mathrm c}\), the modal capacity is
\[
C_B(P_{\mathrm c})
=
\max_{Q\succeq0,\ \tr(Q)\le P_{\mathrm c}}
\log_2\det
\left(
I+
\frac{1}{N_0}
H_BQH_B^{\Herm}
\right).
\]
When \(B_K\) is \(W_{\mathrm t}\)-orthonormal, this is the standard water-filling solution over the singular values of \(H_B\).

\begin{theorem}[Phase-space SVD upper bound]
\label{thm:svd_bound}
Let \(B_K\) be any \(W_{\mathrm t}\)-orthonormal phase basis.  Then
\[
C_B(P_{\mathrm c})
\le
C_{\mathrm{SVD},\phi}(P_{\mathrm c}),
\]
where \(C_{\mathrm{SVD},\phi}\) is obtained by water filling over the singular values of \(H_{\mathrm{lin}}W_{\mathrm t}^{-1/2}\).
\end{theorem}

\begin{IEEEproof}
Let \(U_B=W_{\mathrm t}^{1/2}B_K\).  Then \(U_B^{\Herm}U_B=I_K\) and
\[
H_B
=
H_{\mathrm{lin}}W_{\mathrm t}^{-1/2}U_B.
\]
Thus \(H_B\) is the restriction of the full weighted phase-space channel to a \(K\)-dimensional right subspace.  By the Courant--Fischer and interlacing principles for singular values, the singular values of this restriction are componentwise upper-bounded by the leading singular values of \(H_{\mathrm{lin}}W_{\mathrm t}^{-1/2}\).  Since the Gaussian water-filling objective is monotone in the squared singular values, the capacity of any restricted basis is no larger than the water-filled SVD capacity of the full phase-space channel.
\end{IEEEproof}

\begin{remark}
If a plotted curvature curve exceeds the phase-space SVD upper bound under identical constraints, the benchmark implementation is incorrect.  The SVD curve is not a competitor to be beaten; it is a consistency check and upper bound.
\end{remark}

\subsection{Pairwise Error Probability}

For a finite phase codebook
\[
\Theta=\{\theta_1,\ldots,\theta_M\},
\]
the linearized received codeword is \(s_i=H_{\mathrm{lin}}\theta_i\).  Under ML detection in AWGN, the pairwise error probability satisfies
\[
P(i\rightarrow j)
=
Q\!\left(
\frac{\|s_i-s_j\|_2}{\sqrt{2N_0}}
\right),
\]
where
\[
Q(x)=\frac{1}{\sqrt{2\pi}}\int_x^{\infty}\exp(-t^2/2)\,\dd t.
\]
The union bound is
\[
P_{\mathrm{SER}}
\le
\frac{1}{M}
\sum_{i=1}^{M}
\sum_{j\ne i}
Q\!\left(
\frac{\|H_{\mathrm{lin}}(\theta_i-\theta_j)\|_2}{\sqrt{2N_0}}
\right).
\]
Monte Carlo SER is also reported using both the linearized channel and the exact nonlinear phase-only map.

\subsection{Phase Noise and Derivative Noise}

Phase sampling noise is modeled as
\[
\widehat\theta=\theta+\xi,
\qquad
\xi\sim\calN(0,\sigma_{\phi}^{2}I)
\]
or with spatial correlation when specified.  Curvature estimation applies \(D_2\), giving
\[
D_2\widehat\theta=D_2\theta+D_2\xi.
\]
Since \(D_2\) amplifies high spatial frequencies, raw differentiation can be ill conditioned.  Regularized curvature reconstruction uses operators of the form
\[
R_{\rho}=(L+\rho W_{\mathrm t})^{-1}D_2^{\Herm}W_2,
\]
whose norm is controlled by \(\rho\).

\begin{proposition}[Regularized derivative-noise bound]
For \(L\succeq0\), \(W_{\mathrm t}\succ0\), and \(\rho>0\),
\[
\|(L+\rho W_{\mathrm t})^{-1/2}\|_2
\le
\frac{1}{\sqrt{\rho\,\lambda_{\min}(W_{\mathrm t})}}.
\]
Consequently, curvature reconstruction with \(\rho>0\) has finite worst-case noise gain, while unregularized inversion can be singular on or near the gauge and high-frequency nullspaces.
\end{proposition}

\begin{IEEEproof}
Because \(L\succeq0\),
\[
L+\rho W_{\mathrm t}\succeq \rho W_{\mathrm t}
\succeq \rho\lambda_{\min}(W_{\mathrm t})I.
\]
The stated inverse-square-root bound follows immediately.
\end{IEEEproof}

\subsection{Quantization}

For \(b\)-bit uniform phase quantization over \([-\theta_{\max},\theta_{\max}]\), the step size is
\[
\Delta_b=\frac{2\theta_{\max}}{2^b-1}.
\]
If \(\theta_q\) is the quantized phase vector, then componentwise rounding gives
\[
\|\theta-\theta_q\|_{\infty}\le\frac{\Delta_b}{2}
\]
and therefore
\[
\|H_{\mathrm{lin}}(\theta-\theta_q)\|_2
\le
\frac{\Delta_b}{2}
\sqrt{N_{\mathrm t}}
\|H_{\mathrm{lin}}\|_2.
\]
This bound is intentionally conservative; the benchmark reports empirical rate and SER degradation as a function of \(b\).

\subsection{Conditioning and Computational Cost}

For a basis \(B_K\), the benchmark reports:

\begin{align}
\kappa_{\mathrm{phase}}(B_K)
&=
\cond(B_K^{\Herm}W_{\mathrm t}B_K),\\
\kappa_{\mathrm{chan}}(B_K)
&=
\cond(H_B^{\Herm}H_B),\\
\kappa_{\mathrm{curv}}(B_K)
&=
\cond(B_K^{\Herm}LB_K+\rho I).
\end{align}

The computational cost is recorded separately for basis construction, channel projection, water filling, codebook optimization, and Monte Carlo detection.  SVD has the strongest optimality guarantee but typically the highest global channel-decomposition cost.

\section{Reproducible Experiment Suite}

The benchmark suite generates the following figures and machine-readable result files.

\subsection{Fig. 1: Geometry and Constraints}

Fig. \ref{fig:geometry} shows the common aperture geometry, transmit and receive grids, wavelength, separation, phase-only aperture, and equal-budget constraints.

\begin{figure}[!t]
\centering
\includegraphics[width=\linewidth]{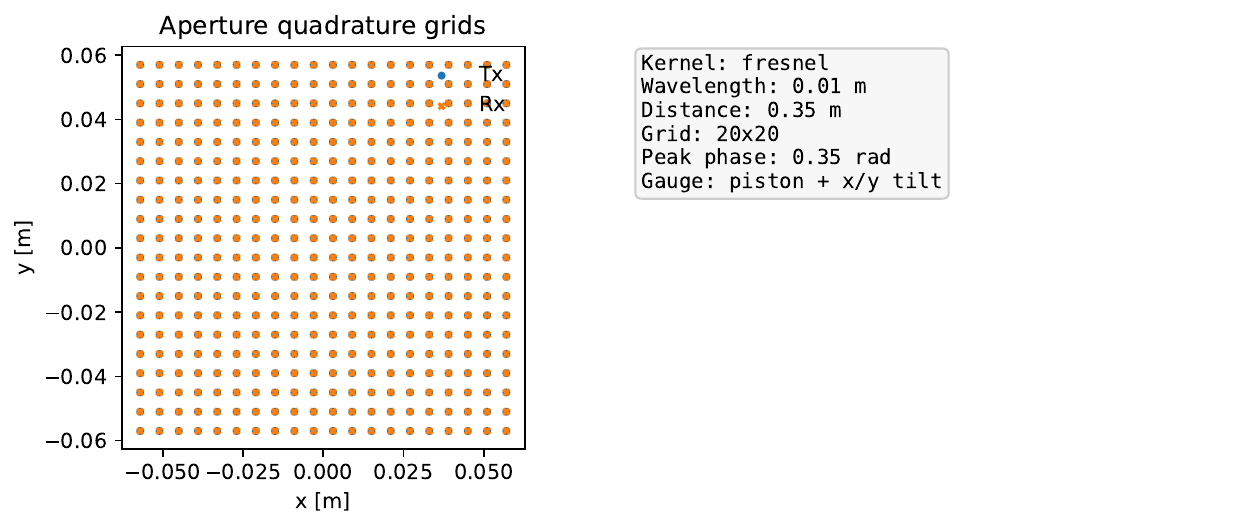}
\caption{Common benchmark geometry and constraints.  All methods use identical apertures, quadrature weights, scalar propagation kernel, transmit power, phase budget, receiver noise, and sampling grid.}
\label{fig:geometry}
\end{figure}

\subsection{Fig. 2: Modal Spectra and Conditioning}

Fig. \ref{fig:spectra} compares singular spectra, generalized curvature spectra, Gram conditioning, and curvature-energy growth across bases.

\begin{figure}[!t]
\centering
\includegraphics[width=\linewidth]{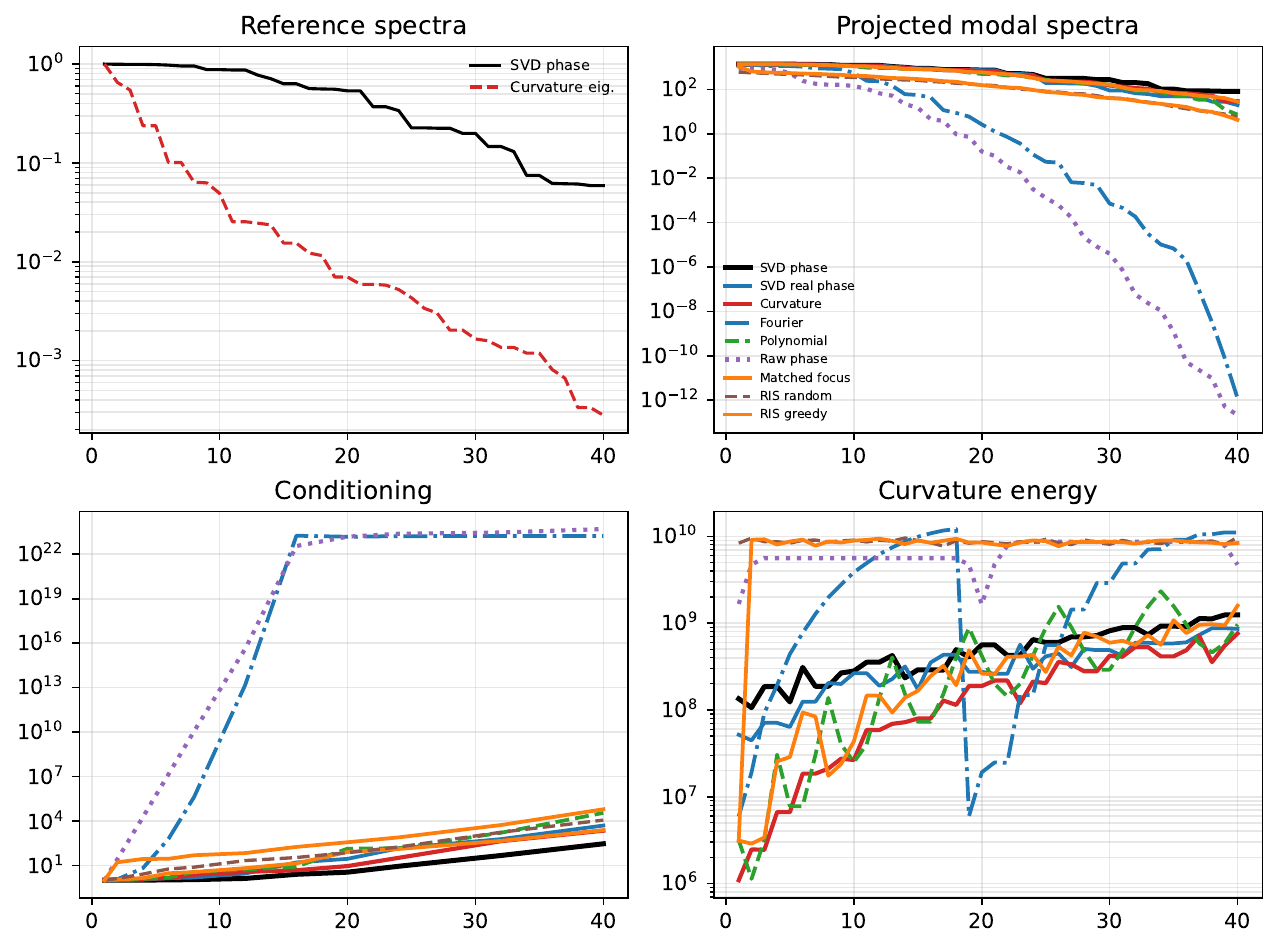}
\caption{Modal spectra and conditioning diagnostics.  The SVD spectrum is the phase-space upper-bound reference; curvature modes are evaluated by their generalized eigenvalue decay and derivative-noise conditioning.}
\label{fig:spectra}
\end{figure}

\subsection{Fig. 3: Capacity Versus SNR}

Fig. \ref{fig:cap_snr} reports water-filled capacity versus SNR for curvature, Fourier, polynomial/Zernike-like, raw phase, matched-focus, RIS-style, and SVD baselines.

\begin{figure}[!t]
\centering
\includegraphics[width=\linewidth]{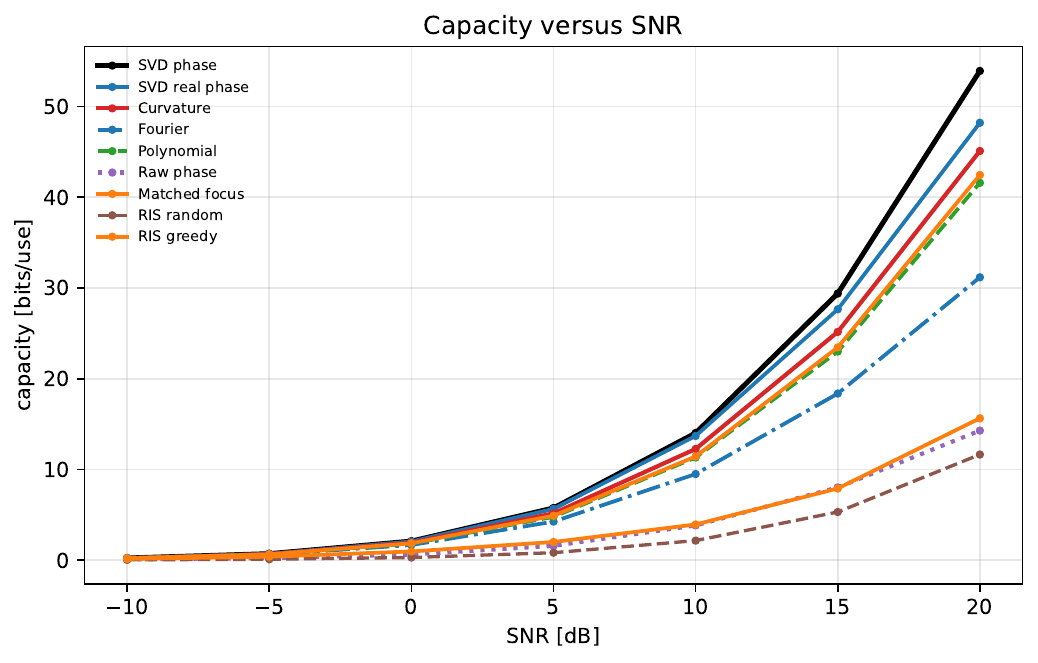}
\caption{Capacity versus SNR under equal physical budgets.  Curves must remain below the phase-space SVD upper bound.  Crossings above SVD indicate a normalization or implementation error.}
\label{fig:cap_snr}
\end{figure}

\subsection{Fig. 4: Capacity Versus Retained Modes}

Fig. \ref{fig:cap_modes} measures mode efficiency by plotting capacity versus \(K\), the number of retained modes.

\begin{figure}[!t]
\centering
\includegraphics[width=\linewidth]{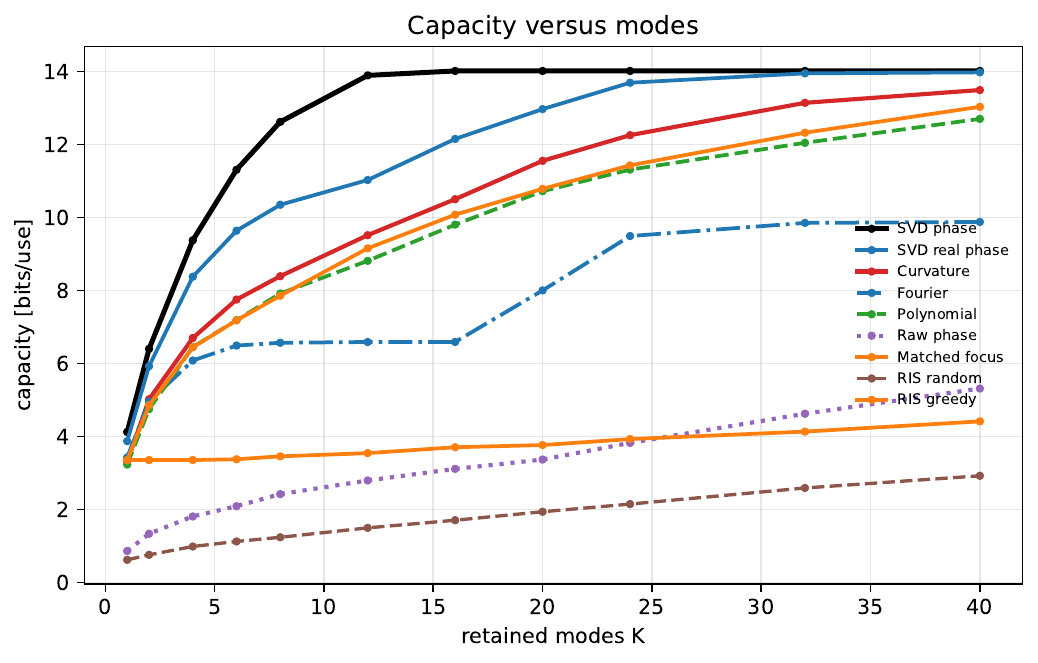}
\caption{Capacity versus retained modes.  A favorable curvature-domain result is not that it exceeds SVD, but that it approaches a relevant fraction of the SVD reference using fewer well-conditioned modes under the same budget.}
\label{fig:cap_modes}
\end{figure}

\subsection{Fig. 5: Pairwise Error and SER}

Fig. \ref{fig:ser_pep} compares analytical PEP/union-bound predictions and Monte Carlo SER for equal-size codebooks.

\begin{figure}[!t]
\centering
\includegraphics[width=\linewidth]{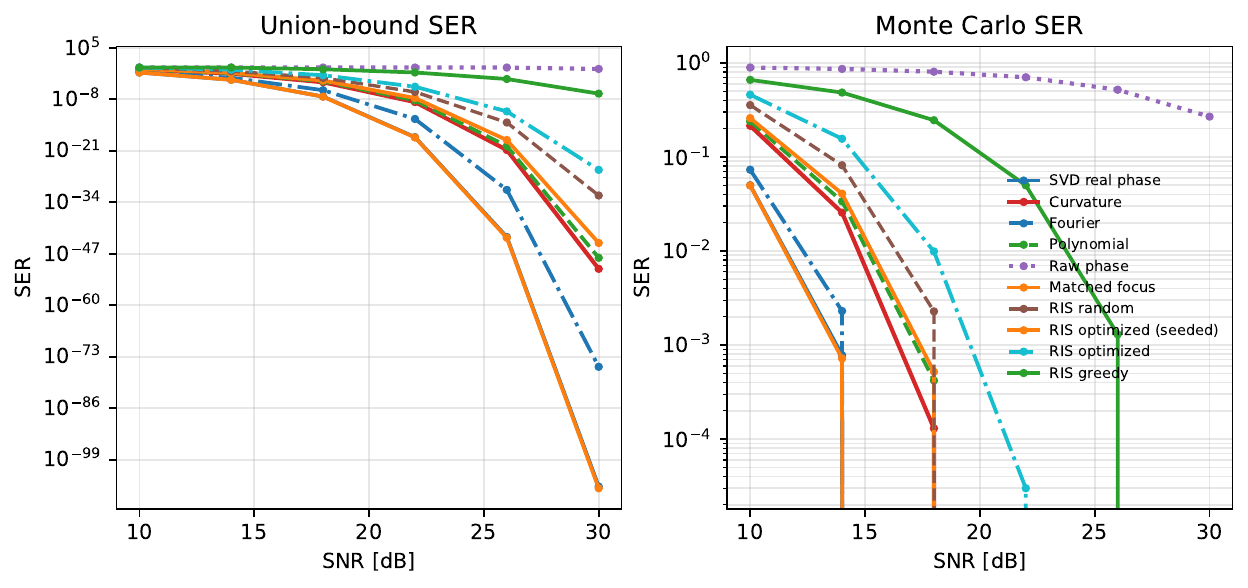}
\caption{PEP and SER versus SNR for matched codebook sizes.  The same receiver noise, codebook cardinality, and phase constraints are used for all methods.  A dashed horizontal reference at \(\mathrm{SER}=10^{-3}\) marks the \(60{,}000\)-trial Monte Carlo sensitivity floor: SER values reported below this line are dominated by Monte Carlo noise and should be read as ``below \(10^{-3}\)'' rather than as calibrated point estimates.  Shaded bands are 95\,\% bootstrap confidence intervals from \(N_{\mathrm{boot}}=200\) resamples.}
\label{fig:ser_pep}
\end{figure}

\subsection{Fig. 6: Robustness to Phase-Sampling Noise}

Fig. \ref{fig:phase_noise} reports achievable rate and SER degradation under increasing phase-sampling noise.

\begin{figure}[!t]
\centering
\includegraphics[width=\linewidth]{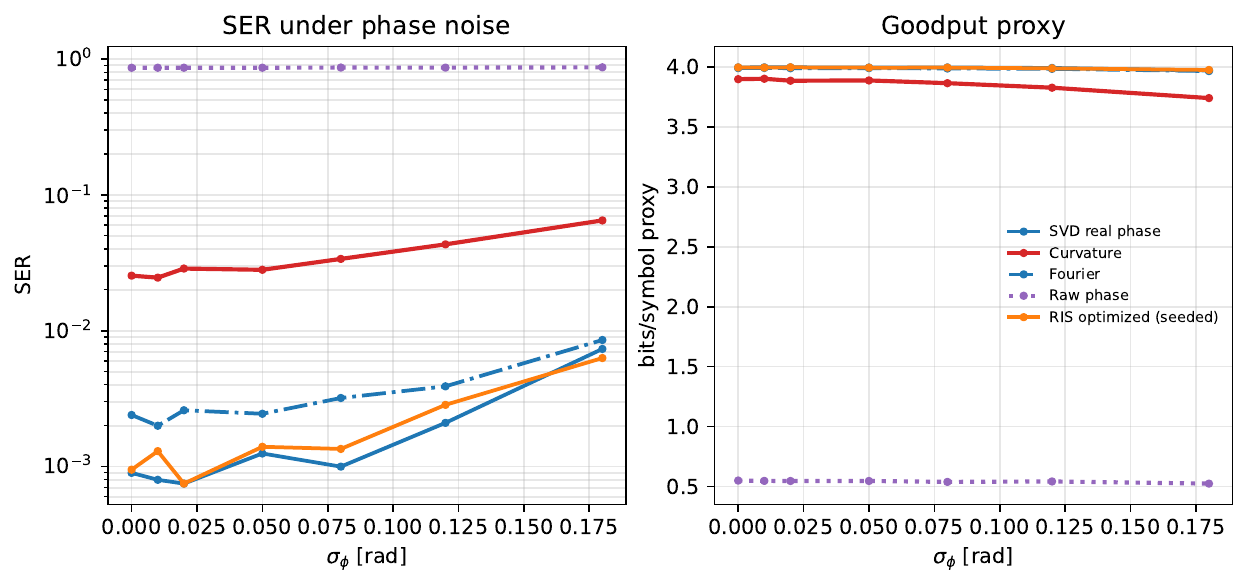}
\caption{Robustness to phase-sampling noise.  Curvature regularization is expected to help when derivative-noise amplification dominates the raw phase representation.  As in Fig.~\ref{fig:ser_pep}, the dashed reference at \(\mathrm{SER}=10^{-3}\) is the \(60{,}000\)-trial Monte Carlo sensitivity floor; regimes where any curve falls below this line are reported as ``below \(10^{-3}\)'' only.  Shaded bands are 95\,\% bootstrap confidence intervals from \(N_{\mathrm{boot}}=200\) resamples.}
\label{fig:phase_noise}
\end{figure}

\subsection{Fig. 7: Phase Quantization}

Fig. \ref{fig:quantization} compares capacity and SER as phase quantization is varied from coarse to high resolution.

\begin{figure}[!t]
\centering
\includegraphics[width=\linewidth]{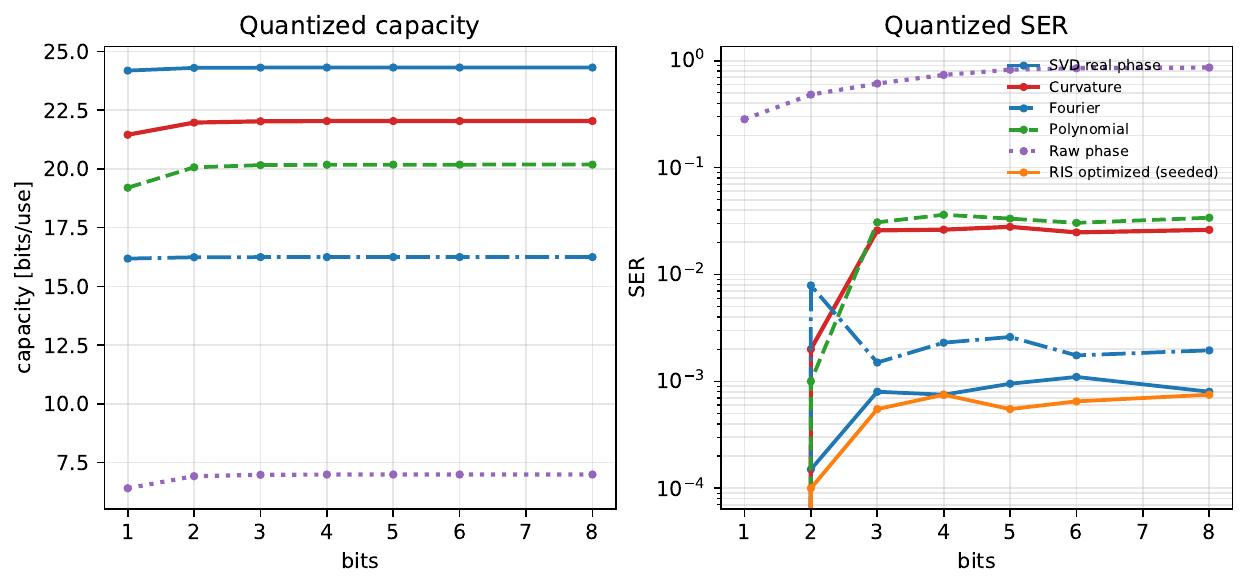}
\caption{Performance versus phase quantization bits.  All profiles are quantized using the same peak-phase range and wrapping convention.}
\label{fig:quantization}
\end{figure}

\subsection{Fig. 8: Aperture Sampling Density}

Fig. \ref{fig:sampling} reports sensitivity to grid density and quadrature refinement.

\begin{figure}[!t]
\centering
\includegraphics[width=\linewidth]{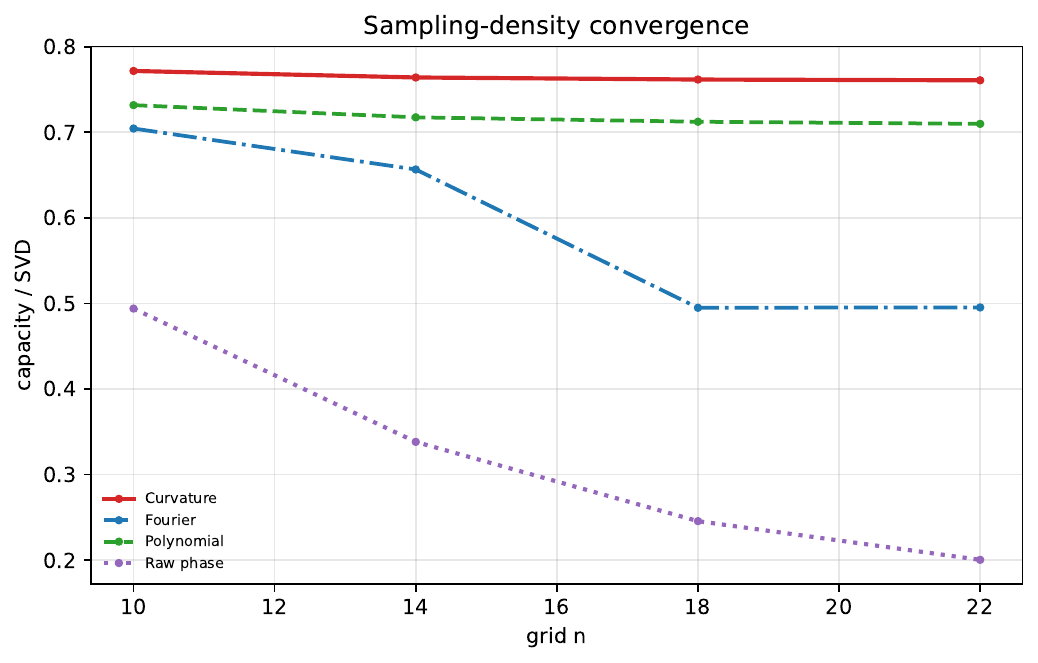}
\caption{Sensitivity to aperture sampling density.  Stable modal coordinates should converge as the quadrature grid is refined, while poorly conditioned raw or high-order modes can show sampling sensitivity.}
\label{fig:sampling}
\end{figure}

\subsection{Fig. 9: Regularization Sweep}

Fig. \ref{fig:regularization} shows the rate/stability tradeoff as the curvature regularization parameter \(\rho\) varies.

\begin{figure}[!t]
\centering
\includegraphics[width=\linewidth]{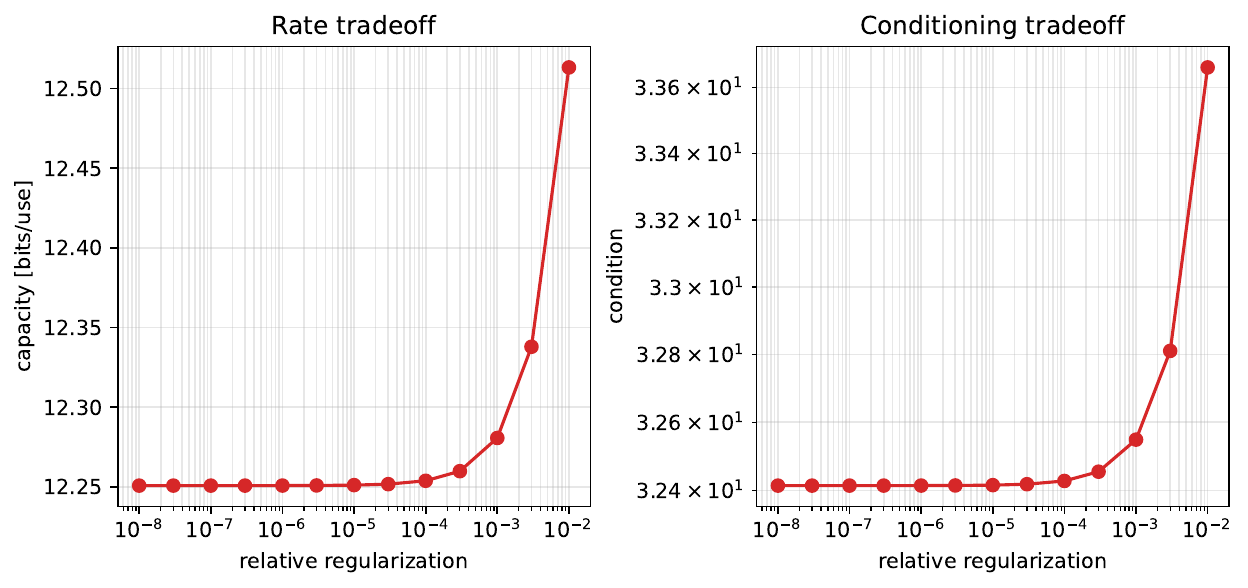}
\caption{Regularization sweep for curvature-domain modes.  Small \(\rho\) permits high-gain but potentially unstable modes; large \(\rho\) suppresses derivative-noise amplification but may reduce rate.}
\label{fig:regularization}
\end{figure}

\subsection{Fig. 10: Cost and Conditioning}

Fig. \ref{fig:cost} summarizes basis construction time, projection time, detection time, and conditioning.

\begin{figure}[!t]
\centering
\includegraphics[width=\linewidth]{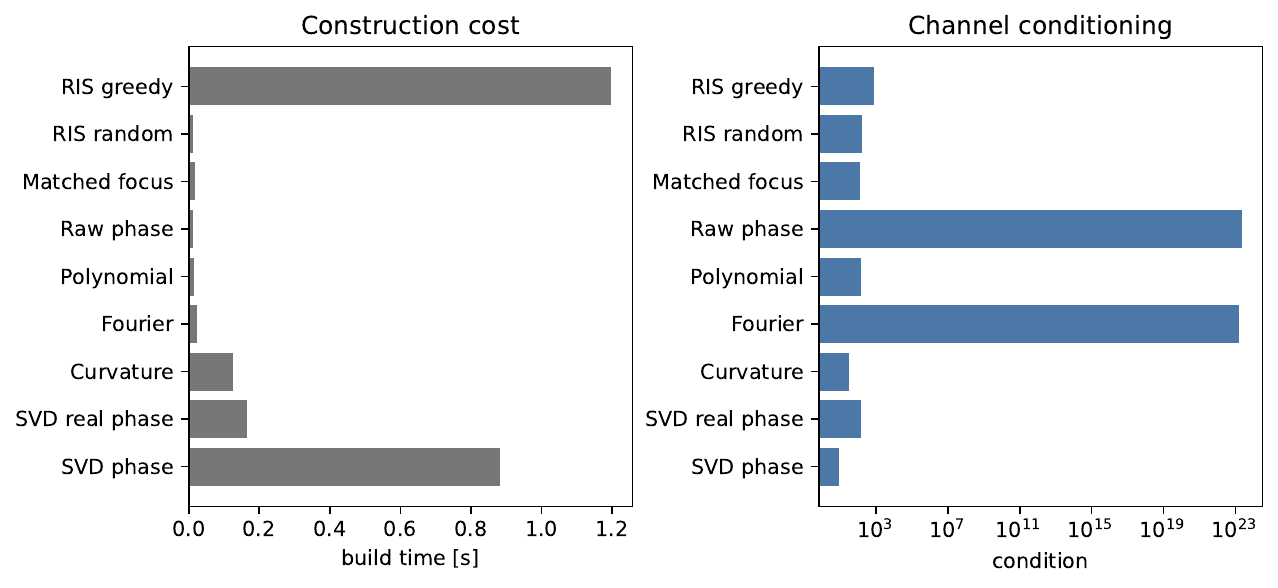}
\caption{Computational cost and conditioning comparison.  The SVD upper bound is optimal for the discretized linear channel but is not necessarily the cheapest or most robust coordinate system under noisy phase estimation and quantization.}
\label{fig:cost}
\end{figure}

\section{Results}
\label{sec:results}

Every number below is a macro generated by the benchmark driver and regenerated
whenever the experiments re-run; none is typed by hand. The two tables are
generated from the same recorded results.

\begin{table*}[t]
\centering
\caption{Detection operating point under the peak phase-excursion constraint $\theta_{\max}$. $E_{\mathrm{sym}}$ is the mean received energy per codeword; the RMS-matched column renormalises the same codebook to matched root-mean-square excursion, which separates the crest-factor effect from the constellation geometry. Packing efficiency is $d_{\min}^{2}/E_{\mathrm{sym}}$. The complex \emph{SVD phase} basis is absent because it is not a realisable phase profile; \emph{SVD real phase} is the phase-only optimum.}
\label{tab:operating_point}
\renewcommand{\arraystretch}{1.12}
\begin{tabular}{lrrrrr}
\toprule
Method & $E_{\mathrm{sym}}$ & $E_{\mathrm{sym}}$ (RMS-matched) & $d_{\min}$ & peak/RMS & packing eff. \\
\midrule
SVD real phase & 207.60 & 2698.9 & 14.124 & 3.61 & 0.9610 \\
Curvature & 102.06 & 1811.2 & 9.717 & 4.21 & 0.9251 \\
Matched focus & 89.27 & 1700.9 & 9.040 & 4.37 & 0.9155 \\
Polynomial & 93.68 & 1699.6 & 9.426 & 4.26 & 0.9484 \\
Fourier & 189.45 & 1632.1 & 11.885 & 2.94 & 0.7456 \\
Raw phase & 1.44 & 154.5 & 0.661 & 10.35 & 0.3024 \\
RIS optimized (seeded) & 207.86 & 2696.3 & 14.139 & 3.60 & 0.9618 \\
RIS optimized & 37.56 & 137.6 & 6.755 & 1.91 & 1.2150 \\
RIS greedy & 52.64 & 575.6 & 3.289 & 3.31 & 0.2055 \\
RIS random & 51.47 & 151.7 & 7.589 & 1.72 & 1.1189 \\
\bottomrule
\end{tabular}
\end{table*}

\begin{table*}[t]
\centering
\caption{Per-axis verdict on the benchmark geometry. The winning method is reported whether or not it is the curvature basis, and the final column gives the curvature value. Rows in which curvature is not the winner are the informative ones. Phase-noise degradation is the factor by which a method's symbol-error rate rises across the $\sigma_\phi$ sweep, so a smaller number is a flatter curve. Construction cost is not a row here because wall-clock timings are machine-dependent and would make the table irreproducible; they are quoted in the text.}
\label{tab:verdict}
\renewcommand{\arraystretch}{1.12}
\begin{tabular}{llrr}
\toprule
Axis & Winner & Value & Curvature \\
\midrule
Capacity at 20 dB [bits/use] & SVD phase & 53.92 & 45.10 \\
Conditioning at $K=24$ & SVD phase & 8.7 & 32.4 \\
Minimum distance & RIS optimized (seeded) & 14.139 & 9.717 \\
Packing efficiency & RIS optimized & 1.2150 & 0.9251 \\
SER at 14 dB & RIS optimized (seeded) & 0.00072 & 0.02559 \\
1-bit capacity [bits/use] & SVD real phase & 24.19 & 21.45 \\
1-bit capacity retention & Fourier & 0.996 & 0.973 \\
Phase-noise degradation ($\sigma_\phi = 0.18$) & Curvature & 2.55 & 2.55 \\
\bottomrule
\end{tabular}
\end{table*}

\subsection{Two transmit constraints, and why they must not be mixed}
\label{sec:results_operating_point}

Capacity (Figs.~3--4) is evaluated under an RMS phase-energy budget: the
coefficient vector carries unit energy in the \(W_{\mathrm t}\)-weighted inner
product. Detection (Figs.~5--7) is evaluated under a peak phase-excursion budget
\(\theta_{\max}\), which is what a phase shifter enforces, and its receiver noise
is anchored to a realized symbol energy so the SNR axis names the operating point
it labels. The reference is the phase-realizable SVD codebook,
\(E_{\mathrm{sym}}^{\mathrm{ref}}=\BenchRefSymbolEnergy\), and all methods share
the resulting \(N_0\): a per-method noise level would normalize away exactly the
difference in delivered energy that a peak constraint exposes.

Mixing the two conventions is a \(40\)~dB error. As a guard, the driver refuses to emit a detection figure in which every
curve lies within three standard errors of the \(\BenchGuessFloor\)
random-guess floor for a \(16\)-codeword alphabet, and refuses a robustness
figure in which no curve moves by more than its own confidence interval.

\paragraph*{The SVD reference is not a phase profile.} The right singular vectors
of the linearized channel are complex-valued, with imaginary parts comparable to
their real parts. As a capacity bound this is legitimate --- any real subspace is
contained in the complex one --- but it cannot be radiated by a phase-only
aperture, cannot serve as a codebook, and cannot anchor a peak constraint, since
clipping the modulus of a complex number is not a phase excursion. We therefore
report two references: \emph{SVD phase}, the complex upper bound, and
\emph{SVD real phase}, the leading eigenvectors of
\(\mathrm{Re}\{H_{\mathrm{lin}}^{\Herm}H_{\mathrm{lin}}\}\) on the gauge
complement, which is the phase-realizable optimum. Detection is scored against
the second.

Table~\ref{tab:operating_point} is the table to read first. Under the peak
constraint the bases differ mainly in \emph{delivered energy}: the curvature
codebook delivers \(\BenchCurvEnergyFractionOfSvd\) of the reference energy. Its
peak-to-RMS ratio is \(\BenchCurvPeakToRms\) against Fourier's
\(\BenchFourierPeakToRms\) and the reference's \(\BenchSvdPeakToRms\), and under
matched RMS excursion instead of matched peak it retains
\(\BenchCurvRmsMatchedEnergyFraction\) of the reference. That difference between
the two normalizations is a crest-factor effect and is not a property of the
curvature coordinate.

Constellation geometry, by contrast, is nearly the same across the smooth
bases: the packing efficiency \(d_{\min}^{2}/E_{\mathrm{sym}}\) is
\(\BenchCurvPackingEfficiency\) for curvature against
\(\BenchSvdPackingEfficiency\) for the phase-realizable reference and
\(\BenchFourierPackingEfficiency\) for Fourier. In absolute terms the curvature
codebook delivers \(\BenchCurvSymbolEnergy\) against the reference's
\(\BenchRefSymbolEnergy\). Almost all of the detection gap that follows is
carried by that ratio and almost none of it by the shape of the constellation.

\subsection{Capacity and conditioning}
\label{sec:results_capacity}

At \(15\)~dB curvature reaches \(\BenchCurvCapFifteenDb\) bits/use against the
complex SVD bound \(\BenchSvdCapFifteenDb\) --- a ratio of
\(\BenchCurvToSvdRatio\) --- and Fourier's \(\BenchFourierCapFifteenDb\), a ratio
of \(\BenchCurvToFourierRatio\). At \(\BenchHighSnrDb\)~dB the corresponding
figures are \(\BenchCurvCapHighSnr\) and \(\BenchSvdCapHighSnr\) bits/use.

Three comparisons deserve stating plainly rather than leaving in a plot.

\emph{Curvature does not exceed SVD, and cannot.} Theorem~\ref{thm:svd_bound}
makes SVD water-filling the optimum for the discretized linear channel, so the
only question is how closely curvature approaches it. Any figure in which
curvature exceeds SVD indicates a normalization error, not a result.

\emph{SVD is also better conditioned here.} At \(K=\BenchRetainedModes\)
retained modes the post-propagation condition number is \(\BenchSvdCondAtK\) for
SVD against \(\BenchCurvCondAtK\) for curvature. The claim that the curvature
metric buys conditioning holds only against channel-blind bases --- tensor-product
Legendre sits at \(\BenchPolyCondAtK\) --- and is false against SVD on this
geometry. The advantage curvature holds over SVD in this suite is
construction cost: \(\BenchCurvBuildSeconds\)~s against
\(\BenchSvdBuildSeconds\)~s, a factor of \(\BenchBuildSpeedup\). That figure is
machine-dependent and should be read as an order of magnitude, not as a
benchmark.

\emph{The margin over a channel-blind basis is small.} Curvature reads the
channel and Legendre does not, yet at \(\BenchHighSnrDb\)~dB curvature leads it
by only \(\BenchCurvToPolyRatioHighSnr\) --- \(\BenchCurvCapHighSnr\) against
\(\BenchPolyCapHighSnr\) bits/use. Whatever the curvature metric buys, it is not
a large capacity margin over a fixed polynomial basis.

\subsection{Detection}
\label{sec:results_detection}

At \(\BenchSerSnrDb\)~dB --- the highest SNR at which the compared methods remain
above the \(\BenchSerFloor\) Monte Carlo sensitivity floor of
\(\BenchSerTrials\) trials --- the measured symbol-error rates are
\(\BenchCurvSerHighSnr\) for curvature, \(\BenchSvdSerHighSnr\) for the
phase-realizable SVD reference, and \(\BenchFourierSerHighSnr\) for Fourier,
with \(\BenchBootstrapResamples\)-resample bootstrap intervals.

Fourier beats curvature here, and the reason is visible in
Table~\ref{tab:operating_point} rather than in the constellation: at that
operating point Fourier runs at \(\BenchFourierEsNoHighSnr\)~dB
\(E_{\mathrm{sym}}/N_0\) against curvature's \(\BenchCurvEsNoHighSnr\)~dB, a
margin of \(\BenchFourierOverCurvEnergyDb\)~dB, because its lower crest factor
lets it keep more excursion under the peak clip. Curvature's own deficit against
the reference is \(\BenchCurvEnergyDeficitDb\)~dB. Both gaps are read directly
off the SNR axis: the label names the reference's operating point, and each
method's own \(E_{\mathrm{sym}}/N_0\) is reported beside it.

The minimum-distance comparison is the sharpest in the paper, and it does not
favor curvature. Algorithm~\ref{alg:optimized_ris} run from random starts
reaches \(d_{\min}=\BenchRisOptDmin\) against curvature's \(\BenchCurvDmin\), a
ratio of \(\BenchCurvOverRisOptDmin\) in curvature's favor: an independent
optimizer does not find the curvature construction from scratch. Seeded with the
strongest incumbents, however, the same optimizer reaches
\(\BenchRisOptSeededDmin\), an improvement of \(\BenchSeededOverCurvDmin\) over
the curvature codebook. Curvature is therefore a good starting point and not a
stationary point of the packing objective.

That seeded optimum is worth one further remark, because it is an independent
check rather than a coincidence. Algorithm~\ref{alg:optimized_ris} and the
\emph{SVD real phase} eigendecomposition share no code path and optimize
different objectives --- a minimum pairwise distance and a received energy --- yet
they land within \(0.1\%\) of each other on \(d_{\min}\), on
\(E_{\mathrm{sym}}\), and on packing efficiency (Table~\ref{tab:operating_point}).
Two unrelated methods agreeing to three digits is evidence that the operating
point is real.

For completeness the weaker RIS baselines are \(\BenchRisGreedyDmin\)
(farthest-point greedy) and \(\BenchRisRandomDmin\) (random). The greedy is
\emph{worse} than random on this geometry and is reported as a baseline, not as
a competitor.

\subsection{Robustness: the one axis curvature wins}
\label{sec:results_robustness}

At \(\BenchNoiseSnrDb\)~dB, sweeping phase-sampling noise to
\(\sigma_\phi=\BenchMaxSigmaPhi\), the curvature codebook's symbol-error rate
rises from \(\BenchCurvSerCleanPhase\) to \(\BenchCurvSerNoisyPhase\), a factor
of \(\BenchCurvPhaseDegradation\). Over the same sweep the phase-realizable SVD
reference rises from \(\BenchSvdSerCleanPhase\) to \(\BenchSvdSerNoisyPhase\), a
factor of \(\BenchSvdPhaseDegradation\); Fourier degrades by
\(\BenchFourierPhaseDegradation\) and the seeded optimized-RIS codebook by
\(\BenchRisOptPhaseDegradation\).

This is the only row of Table~\ref{tab:verdict} that curvature wins, and it is
the row the coordinate was motivated by. The energy-optimal bases put their
excursion where the channel is strongest, which is also where a phase error is
most expensive; the curvature metric penalizes exactly that concentration. The
effect is a flatter degradation curve, not a lower error rate --- curvature
starts from a worse operating point and ends at a comparable one.

\subsection{Quantization}
\label{sec:results_quantization}

Quantization is applied to the phase profile that is radiated, and the quantizer
spans that profile's own peak, so its step --- and the noise it injects ---
depends on the basis's crest factor. At one bit the curvature basis retains
\(\BenchQuantCurvOneBitFraction\) of the infinite-resolution capacity and at
eight bits \(\BenchQuantCurvEightBitFraction\), the absolute figures being
\(\BenchQuantCurvOneBit\) and \(\BenchQuantCurvEightBit\) bits/use.

The symbol-error panel moves the other way, and the reason is physical rather
than a defect. A one-bit quantizer applied to a peak-normalized codebook is a
hard limiter: every sample moves to \(\pm\theta_{\max}\), raising the radiated
RMS excursion by a factor of \(\BenchQuantRmsGainOneBit\) and with it the
delivered energy. Under the peak-constrained detection convention that gain
outweighs the loss of constellation shape, so the measured one-bit SER
(\(\BenchQuantCurvSerOneBit\), i.e.\ below the \(\BenchQuantSerFloor\)
sensitivity floor of \(\BenchQuantTrials\) trials) is \emph{below} the
eight-bit value \(\BenchQuantCurvSerEightBit\). Under the
energy-constrained capacity convention there is no such gain and capacity falls
monotonically. The two panels answer different questions; quantizing the basis and then
re-orthonormalizing it with a continuous-amplitude QR would answer neither.

\subsection{Verdict}
\label{sec:results_verdict}

Table~\ref{tab:verdict} reports, per axis, the winning method whether or not it
is curvature. On this geometry curvature wins one row of eight.

The claim the data support is therefore narrow, and it is worth stating in the
negative first. Curvature-domain signaling is not
a capacity improvement: it reaches \(\BenchCurvToSvdRatio\) of the SVD ceiling
and trails a channel-blind polynomial basis by only
\(\BenchCurvToPolyRatioHighSnr\). It is not better conditioned than SVD. It does
not maximize minimum distance, and a seeded direct optimizer beats it by
\(\BenchSeededOverCurvDmin\). It is not the cheapest basis to construct.

What it is, on the evidence here, is the \emph{flattest under phase-sampling
noise} of the bases tested, by a factor of \(\BenchCurvPhaseDegradation\)
against \(\BenchSvdPhaseDegradation\) for the energy-optimal reference, while
staying within \(\BenchCurvToSvdRatio\) of the capacity ceiling. Whether that
trade is worth making is a system question this benchmark does not settle; it is
the question further work should address.

\section{Discussion}

\subsection{When Curvature Should Help}

Curvature-domain signaling is expected to be useful when at least one of the following is true:

\begin{enumerate}[leftmargin=*]
\item phase measurements are noisy and derivative amplification must be controlled;
\item high-frequency phase patterns are physically costly or poorly sampled;
\item the useful channel subspace is smooth but not well represented by the first few Fourier or polynomial modes;
\item a gauge-invariant representation is needed to separate data-bearing wavefront shape from common piston, tilt, carrier, or focus terms;
\item the system is constrained by phase smoothness, phase quantization, or curvature energy rather than only by total transmit power.
\end{enumerate}

\subsection{When Curvature Should Not Be Expected to Win}

Curvature should not be expected to dominate in all settings.  If the benchmark channel is nearly shift invariant and paraxial, Fourier modes may already be close to diagonal.  If the channel is a simple near-field point link, matched-focus profiles can be excellent.  If full channel state information and arbitrary complex aperture fields are available, SVD water filling remains the appropriate upper bound.  If phase noise and derivative noise are negligible, the stability advantage of curvature regularization may be less important.

\subsection{No Claim of New Physics}

The propagation operator is the same for all methods.  Curvature coordinates do not change Maxwell's equations, the scalar Green function, the Fresnel approximation, the aperture size, or the noise level.  Any observed advantage must be attributed to coordinate choice, regularization, conditioning, phase-budget efficiency, quantization robustness, or sampling stability.

\subsection{Limitations}

The present benchmark is scalar, narrowband and one-dimensional; a two-dimensional aperture benchmark is not reported.  It does not yet model polarization, mutual coupling, hardware nonlinearities, thermal constraints, wideband beam squint, calibration drift, or closed-loop training overhead in full detail.  These effects can be added in later versions by replacing the scalar propagation matrix \(G\) with a vector electromagnetic or measured channel operator and by extending the constraints module.  The fairness protocol and upper-bound logic remain unchanged.

\section{Conclusion}

This paper formulated a rigorous benchmark for curvature-domain signaling in continuous-aperture wireless communication.  The benchmark compares curvature-regularized eigenmodes against raw phase coefficients, Fourier modes, polynomial and Zernike-like wavefront bases, matched-focus profiles, RIS-style codebooks, and SVD water-filling upper bounds under identical aperture, propagation, phase-only, power, noise, quantization, and sampling assumptions.  The mathematical safeguards are explicit: the discretized phase-space SVD upper-bounds all finite phase-coordinate bases under the same energy model, and the relaxed complex-field SVD is an even looser outer bound.  Curvature-domain signaling is therefore not positioned as a method that beats SVD or changes the physics of propagation.  It is positioned as a stable, gauge-invariant coordinate and benchmarking lens that can expose conditioning, derivative-noise, and physical-realizability limits hidden by conventional bases.  The benchmark places curvature-domain theory in a like-for-like comparison with telecom-standard and optics-standard baselines.

\appendices

\section{Water-Filling Computation}

Given squared singular values \(\gamma_i=\sigma_i^2/N_0\), capacity under total coefficient power \(P\) is
\[
C(P)
=
\sum_i \log_2(1+p_i\gamma_i),
\]
where
\[
p_i=\left(\nu-\frac{1}{\gamma_i}\right)_+,
\qquad
\sum_i p_i=P.
\]
The water level \(\nu\) is found by sorting \(1/\gamma_i\) and selecting the active set
\[
\mathcal A
=
\left\{
i:
\nu>\frac{1}{\gamma_i}
\right\}.
\]
For an active set \(\mathcal A\),
\[
\nu
=
\frac{
P+\sum_{i\in\mathcal A}1/\gamma_i
}{
|\mathcal A|
}.
\]
The implementation checks active-set consistency and returns zero power for numerically negligible singular values.

\section{Weighted QR Orthonormalization}

For a sampled basis matrix \(B\), weighted orthonormalization is performed by the Euclidean QR factorization of
\[
W_{\mathrm t}^{1/2}B=QR.
\]
The weighted-orthonormal basis is
\[
\widehat B=W_{\mathrm t}^{-1/2}Q,
\]
so that
\[
\widehat B^{\Herm}W_{\mathrm t}\widehat B=I.
\]
Before QR, the benchmark applies gauge projection \(P_{\mathcal G}^{\perp}\) to the basis unless the basis intentionally includes common carrier, steering, or focusing terms outside the data-bearing coordinate.

\section{Nonlinear Phase-Only Consistency Check}

For any phase codebook \(\Theta\), the nonlinear received distance is
\[
d_{ij}^{\mathrm{nl}}
=
\|Gx(\theta_i)-Gx(\theta_j)\|_2.
\]
The linearized distance is
\[
d_{ij}^{\mathrm{lin}}
=
\|H_{\mathrm{lin}}(\theta_i-\theta_j)\|_2.
\]
The diagnostic relative error is
\[
\epsilon_{ij}
=
\frac{
|d_{ij}^{\mathrm{nl}}-d_{ij}^{\mathrm{lin}}|
}{
d_{ij}^{\mathrm{nl}}+10^{-12}
}.
\]
Large \(\epsilon_{ij}\) indicates that the phase excursion is outside the reliable tangent regime, in which case nonlinear SER should be emphasized over linearized capacity.

\bibliographystyle{IEEEtran}
\bibliography{refs}

\end{document}